\documentclass[12pt,a4paper]{article}
\usepackage{amssymb}
\usepackage{amsmath}
\usepackage[top=1in, bottom=1in, left=1in, right=1in, a4paper]{geometry}
\usepackage[onehalfspacing]{setspace}
\usepackage{amsfonts}
\usepackage{amstext}
\usepackage{amsthm}
\usepackage[round]{natbib}
\usepackage{hyperref}
\usepackage{graphicx}
\usepackage{url}

\numberwithin{equation}{section}
\newtheorem{theorem}{Theorem}

\newtheorem{condition}{Condition}

\newtheorem{lemma}{Lemma}
\theoremstyle{remark}

\input{tcilatex}

\begin{document}

\title{High Dimensional Classification through $\ell _{0}$-Penalized
Empirical Risk Minimization\thanks{{\footnotesize This work was
supported in part by the Ministry of Science and Technology, Taiwan
(MOST105-2410-H-001-003-), Academia Sinica (AS-CDA-106-H01),  the European
Research Council (ERC-2014-CoG-646917-ROMIA), and the UK Economic and Social
Research Council (ESRC) through research grant (ES/P008909/1) to the CeMMAP.}%
}}
\author{Le-Yu Chen\thanks{%
E-mail: lychen@econ.sinica.edu.tw} \\
{\small {Institute of Economics, Academia Sinica}} \and Sokbae Lee\thanks{%
E-mail: sl3841@columbia.edu} \\
{\small {Department of Economics, Columbia University}}\\
{\small {Centre for Microdata Methods and Practice, Institute for Fiscal
Studies} }}
\date{23 November 2018}
\maketitle

\begin{abstract}
We consider a high dimensional binary classification problem and construct a classification procedure by minimizing the empirical misclassification risk with a penalty on the number of selected features. We derive non-asymptotic probability bounds on the estimated sparsity as well as on the excess misclassification risk. In particular, we show that our method yields a sparse solution whose $\ell_0$-norm can be arbitrarily close to true sparsity with high probability and obtain the rates of convergence for the excess misclassification risk. The proposed procedure is implemented via  the method of mixed integer linear programming. Its numerical performance is illustrated in Monte Carlo experiments.

\bigskip

\noindent
\textbf{Keywords}: feature selection,
penalized estimation, mixed integer
optimization, finite sample property

\bigskip

\pagebreak
\end{abstract}

\section{Introduction\label{Sec:Introduction}}

Binary classification is concerned with learning a binary classifier that
can be used to categorize objects into one of two predefined statuses. It
arises in a wide range of applications and has been extensively studied in
the statistics and machine learning literature. For comprehensive surveys
and discussions on the binary classification methods, see e.g. %
\citet{DGL:1996}, \citet{vapnik2000}, \citet{Lugosi:2002}, %
\citet{boucheron2005} and \citet{hastie2009}. Solving for the optimal binary
classifier by minimizing the empirical misclassification risk is known as an
empirical risk minimization (ERM) problem. 
There has been a massive research
interest in high dimensional classification problems where the dimension of
the feature vector used to classify the object's label can be comparable
with or even larger than the available training sample size. It is known
(see e.g., \citet{bickel2004} and \citet{fan2008}) that working directly
with a high dimensional feature space can result in poor classification
performance. To overcome this, it is often assumed that only a small subset
of features are important for classification and feature selection is
performed to mitigate the high dimensionality problem. See \citet{fan2011}
for an overview on the issues and methods for high dimensional
classification.

In this paper, we study the ERM based binary classification in the setting
with high dimensional vectors of features. We propose an $\ell _{0}$-penalized ERM procedure for classification by minimizing over a class of
linear classifiers the empirical misclassification risk with a penalty on
the number of selected features.
Here, the $\ell _{0}$-norm of a real vector refers to the number of non-zero
components of the vector.
 When the Bayes classifier, which is the optimal
classifier that minimizes the population misclassification risk, is also of
the linear classifier form and respects a sparsity condition, we show that
this penalized ERM classification approach can yield a sparse solution for
feature selection with high probability. Moreover, we derive non-asymptotic
bound on the excess misclassification risk and establish its rate of
convergence.

There are alternative ERM based approaches to the high dimensional binary
classification problem. \citet{Greenshtein2006}, \citet{Jiang:10} and %
\citet{chen2018} studied the best subset variable selection approach where
the ERM problem is solved subject to a constraint on a pre-specified maximal
number of selected features. \citet{Jiang:10} further showed that the best
subset ERM problem can be approximated by the $\ell _{0}$-penalized ERM
problem. However, they did not establish theoretical results characterizing
the size of the subset of feature variables selected under the $\ell _{0}$-penalized ERM approach. Neither did they provide numerical algorithms
for solving the $\ell _{0}$-penalized estimation problem.

In the present paper, we take the $\ell _{0}$-penalized ERM approach
and develop a computational method for implementation. We show that there is
a high probability in large samples that the resulting number of selected
features under this penalized estimation approach can be capped above by an
upper bound which can be made arbitrarily close to the unknown smallest
number of features that are relevant for classification. Our penalized ERM
approach is also closely related to the method of structural risk
minimization (see, e.g., \citet[][Chapter 18]{DGL:1996}) where the best
classifier is selected by solving a sequence of penalized ERM problems over
an increasing sequence of spaces of classifiers with the penalty depending
on the complexity of the classifier space measured in terms of the
Vapnik-Chervonenkis (VC) dimension. As will be discussed later, our approach
can also be interpreted in a similar fashion yet with a different type of
complexity penalty.

For implementation, we show that the $\ell _{0}$-penalized ERM problem
of this paper can be equivalently reformulated as a mixed integer linear
programming (MILP) problem. This reformulation enables us to employ modern
efficient MIO solvers to solve our penalized ERM problem. Well-known
numerical solvers such as CPLEX and Gurobi can be used to effectively solve
large-scale MILP problems. See \citet{Nemhauser1999} and %
\citet{bertsimas2005} for classic texts on the MIO theory and applications.
See also \citet{junger2009}, \citet{achterberg2013} and 
\citet[Section
2.1]{bertsimas2016} for discussions on computational advances in solving the
MIO problems.

The present paper is organized as follows. In Section \ref{Sec:ERM}, we
describe the binary classification problem and set forth the $\ell _{0}$-penalized ERM approach. In Section \ref{Sec:Theoretical Properties},
we establish theoretical properties of the proposed classification approach.
In Section \ref{Sec:MIO}, we provide a computational method using the MIO
approach. In Section \ref{Sec:Simulation}, we conduct a simulation study on
the performance of the $\ell _{0}$-penalized ERM approach in high
dimensional binary classification problems. We then conclude the paper in
Section \ref{Sec:Conclusions}. Proofs of all theoretical results of the
paper are collated in Appendix \ref{Sec:Appendix}.

\section{An $\ell_0$-Penalized ERM Approach\label{Sec:ERM}}

Let $Y\in \{0,1\}$ be the binary label or outcome of an object and $X$ a $%
\left( p+1\right) $ dimensional feature vector of that object. Write $%
X=(X_{1},\widetilde{X})$, where 
$X_{1}$ is a scalar random variable that  is always included and has a positive effect and 
$\widetilde{X}$ is the $p$ dimensional
subvector of $X$ subject to feature selection.
 For $x\in 
\mathcal{X}$, let 
\begin{equation}
b_{\theta }(x)\equiv 1\left\{ x_{1}+\widetilde{x}^{\prime }\theta \geq
0\right\} ,  \label{predictor}
\end{equation}%
where $\mathcal{X}$ is the support of $X$, $\theta $ is a vector of
parameters, and $1\left\{ \cdot \right\} $ is an indicator function that
takes value 1 if its argument is true and 0 otherwise.

We consider binary classification using linear classifiers of the form (\ref%
{predictor}). Since the condition $X_{1}+\widetilde{X}^{\prime }\theta \geq
0 $ is invariant with respect to any positive scalar that multiplies both
sides of this inequality, working with the classifier (\ref{predictor})
amounts to normalizing the scale by setting the coefficient of $X_{1}$ to be
unity. For any $p$ dimensional real vector $\theta $, let $\left\Vert \theta
\right\Vert _{0}\equiv \sum\nolimits_{j=1}^{p}1\{\theta _{j}\neq 0\}$ be the 
$\ell _{0}$-norm of $\theta $. Assume that the researcher has a training
sample of $n$ independent identically distributed (i.i.d.) observations $%
\left( Y_{i},X_{i}\right) _{i=1}^{n}$ of $(Y,X)$. We allow the dimension $p$
to be potentially much larger than the sample size $n$. We estimate the
coefficient vector $\theta $ by solving the following $\ell _{0}$-penalized minimization problem%
\begin{equation}
\min\nolimits_{\theta \in \Theta }\text{ }S_{n}(b_{\theta })+\lambda
\left\Vert \theta \right\Vert _{0},  \label{penalized ERM}
\end{equation}%
where $\Theta \subset \mathbb{R}^{p}$ denotes the parameter space, and, for
any indicator function $b:\mathcal{X}\mapsto \left\{ 0,1\right\} $, 
\begin{equation}
S_{n}(b)\equiv \frac{1}{n}\sum\nolimits_{i=1}^{n}1\{Y_{i}\neq b(X_{i})\},
\label{Sn}
\end{equation}%
and $\lambda $ is a given non-negative tuning parameter of the penalized
minimization problem.

The function $S_{n}(b)$ is known as the empirical misclassification risk for
the binary classifier $b$. Minimization of $S_{n}(b)$ over the class of
binary classifiers given by (\ref{predictor}) is known as an empirical risk
minimization (ERM) problem. The penalized ERM approach (\ref{penalized ERM})
enforces dimension reduction by attaching a higher penalty to a classifier $%
b_{\theta }$ which uses more object features for classification. Let $%
\widehat{\theta }$ be a solution to the minimization problem (\ref{penalized
ERM}). We shall refer to the resulting classifier (\ref{predictor})
evaluated at $\widehat{\theta }$ as an $\ell _{0}$-penalized ERM classifier.

For any $m\geq 0$, let%
\begin{equation}
\mathcal{B}_{m}\mathcal{\equiv }\left\{ b_{\theta }:\theta \in \Theta
_{m}\right\}  \label{Bq}
\end{equation}%
where 
\begin{equation}
\Theta _{m}\equiv \{\theta \in \Theta :\left\Vert \theta \right\Vert
_{0}\leq m\}.  \label{theta_r}
\end{equation}%
That is, $\mathcal{B}_{m}$ is the class of all linear classifiers in (\ref%
{predictor}) whose $\theta $ vector has no more than $m$ non-zero
components. For $m\in \{0,1,...,p\}$, let 
\begin{equation}
S_{n}^{C}\left( m\right) \equiv \min\nolimits_{b\in \mathcal{B}_{m}}S_{n}(b).
\label{constrained ERM}
\end{equation}%
Then it is straightforward to see that the minimized objective value of the
penalized ERM problem (\ref{penalized ERM}) is equivalent to that of the
problem%
\begin{equation*}
\min\nolimits_{m\in \{0,1,...,p\}}S_{n}^{C}\left( m\right) +\lambda m.
\end{equation*}%
In other words, our  approach is akin to the
method of structural risk minimization as it amounts to solving ERM problems
over an increasing sequence of classifier spaces $\mathcal{B}_{m}$ which
carries a complexity penalty $\lambda m$. In the next section, we will set
forth regularity conditions on the penalty tuning parameter $\lambda $ and
establish theoretical properties for our classification approach.

\section{Theoretical Properties\label%
{Sec:Theoretical Properties}}

In this section, we study theoretical properties of the $\ell _{0}$-penalized ERM classification approach. Let $F$ denote the joint distribution
of $\left( Y,X\right) $. For any indicator function $b:\mathcal{X}\mapsto
\left\{ 0,1\right\} $, let 
\begin{equation}
S(b)\equiv P\left( Y\neq b(X)\right) .  \label{S(a,b,c)}
\end{equation}%
For $x\in \mathcal{X}$, let%
\begin{align}
\eta (x)& \equiv P(Y=1|X=x),  \label{eta_w} \\
b^{\ast }(x)& \equiv 1\left\{ \eta (x)\geq 0.5\right\} .  \label{b_star}
\end{align}%
For any measurable function $f:\mathcal{W\mapsto }\mathbb{R}$, let $%
\left\Vert f\right\Vert _{1}=E\left[ \left\vert f(W)\right\vert \right] $
denote the $L_{1}$-norm of $f$. The functions $\eta $ and $b^{\ast }$ as
well as the $L_{1}$-norm $\left\Vert \cdot \right\Vert _{1}$ depend on the
data generating distribution $F$. It is straightforward to see that, for any
binary classifier $b$,%
\begin{equation}
S(b)-S(b^{\ast })=E\left[ \left\vert 2\eta (W)-1\right\vert \left\vert
b^{\ast }(W)-b(W)\right\vert \right]  \label{excess risk}
\end{equation}%
so that $S(b)$ is minimized at $b=b^{\ast }$. The optimal classifier $%
b^{\ast }$ is known as the Bayes classifier in the classification literature.

We assess the predictive performance of the $\ell _{0}$-penalized ERM
approach by bounding the excess risk%
\begin{equation}
U_{n}\equiv S(b_{\widehat{\theta }})-S(b^{\ast }).  \label{Un}
\end{equation}%
The difference $U_{n}$ is non-negative by (\ref{excess risk}). Hence, a good
classifier will result in a small value of $U_{n}$ with a high probability
and also on average.

We impose the following assumption.

\begin{condition}
\label{sparsity}For every data generating distribution $F$, there is a
non-negative integer $d$, which may depend on $F$, such that $d\leq p$ and $%
b^{\ast }\in \mathcal{B}_{d}$.
\end{condition}

Let $q$ denote the smallest value of non-negative integers $d$ satisfying $%
b^{\ast }\in \mathcal{B}_{d}$. By Condition \ref{sparsity}, such $q$ value
is finite and always exists. Condition \ref{sparsity} implies that the Bayes
classifier $b^{\ast }$ admits a linear threshold crossing structure in the
sense that the equivalence%
\begin{equation*}
\eta (X)\geq 0.5\Longleftrightarrow X_{1}+\widetilde{X}^{\prime }\theta \geq
0
\end{equation*}%
holds almost surely for some $\theta \in \Theta _{q}$, where $q$ can be
interpreted as the sparsity parameter associated with $b^{\ast }$, which is
unknown in this binary classification problem. Moreover, the assumption that 
$q\leq p$ implies that the feature vector $(X_1, \widetilde{X})$ is rich enough to
embody those relevant ones for constructing the Bayes classifier.

For any two real numbers $x$ and $y$, let $x\vee y\equiv \max \{x,y\}$ and $%
x\wedge y\equiv \min \{x,y\}$. For any $x\geq 0$, let $\left\lceil
x\right\rceil $ and $\left\lfloor x\right\rfloor $ respectively denote the
integer ceiling and floor of $x$. We impose the following condition on the
growing rates of $\lambda $ relative to the sample size.

\begin{condition}
\label{lamda}$\lambda =c\sqrt{n^{-1}\ln (p\vee n)}$ for some constant $c>0$.
\end{condition}

Let 
\begin{eqnarray}
m_{0} &\equiv &q\vee (p\wedge \left\lfloor \lambda ^{-1}\right\rfloor ),
\label{m0} \\
r_{n} &\equiv &q\ln (p\vee n).  \label{rn}
\end{eqnarray}%
The estimate $\left\Vert \widehat{\theta }\right\Vert _{0}$ corresponds to
the number of features selected under the $\ell _{0}$-penalized ERM
approach. We now provide a result on the statistical behavior of $\left\Vert 
\widehat{\theta }\right\Vert _{0}$, which sheds lights on the dimension
reduction performance of our penalized estimation approach.

\begin{theorem}
\label{sparsity bound}Assume $q\geq 1$. Given Conditions \ref{sparsity} and %
\ref{lamda}, for all given $\sigma >0$ and $\epsilon \in (0,1)$, there is a
universal constant $M_{\sigma }$, which depends only on $\sigma $, such that 
\begin{equation}
P\left( \left\Vert \widehat{\theta }\right\Vert _{0}>s\right) \leq
j_{0}e^{-\sigma r_{n}}  \label{bound on the l0 norm of theta_hat}
\end{equation}%
where 
\begin{eqnarray}
s &\equiv &(1+\epsilon )q+\epsilon ,  \label{estimated sparsity} \\
j_{0} &\equiv &\left\lceil \frac{\ln \left( m_{0}\right) -\ln \left(
\epsilon \right) }{\left\vert \ln (2\sqrt{M_{\sigma }})-\ln (c)\right\vert }%
\right\rceil ,  \label{j0}
\end{eqnarray}%
provided that the constant $c$ in Condition \ref{lamda} is sufficiently large
such that 
\begin{equation}
c\geq 2\sqrt{M_{\sigma }}\left( 1+\epsilon \right) \epsilon ^{-1},
\label{condition on c}
\end{equation}%
and the inequality%
\begin{equation}
4\left( k+1\right) \ln \left( M_{\sigma }k\ln (p\vee n)\right) \leq k\ln
(p\vee n)+6\left( k+1\right) \ln 2  \label{inequality}
\end{equation}%
holds for any integer $k$ that satisfies%
\begin{equation*}
q\leq k\leq \left[ m_{0}\vee \left\lfloor s\right\rfloor \vee \left( \left(
j_{0}-1\right) q+\left\lfloor \sqrt{m_{0}}\right\rfloor \right) \right]
\wedge p.
\end{equation*}
\end{theorem}

For any fixed $\epsilon \in (0,1)$, we can deduce from Theorem \ref{sparsity
bound} that $P\left( \left\Vert \widehat{\theta }\right\Vert _{0}>s\right)
\longrightarrow 0$ as $p\vee n\longrightarrow \infty $. Moreover, this
theorem implies that our approach is
effective in reducing the feature dimension in the sense that, with high
probability in large samples, the number of selected features is capped
above by the quantity (\ref{estimated sparsity}), which can be made
arbitrarily close to the true sparsity $q$ in the classification problem.
Specifically, if $\epsilon $ turns out to be smaller than $1/\left(
q+1\right) $, the result (\ref{bound on the l0 norm of theta_hat}) implies
that $P\left( \left\Vert \widehat{\theta }\right\Vert _{0}>q+1\right) $
tends to zero exponentially in $r_{n}$.

The next theorem characterizes the predictive performance of the $\ell _{0}$%
-penalized ERM approach.

\begin{theorem}
\label{convergence rate}Under the setup and assumptions stated in Theorem %
\ref{sparsity bound}, the following result holds: 
\begin{equation}
P\left( U_{n}>3\lambda s\right) \leq \left( 1+j_{0}\right) e^{-\sigma r_{n}}.
\label{bound on excess risk}
\end{equation}
\end{theorem}

Theorem \ref{convergence rate} implies that the tail probability of $U_{n}$
decays to zero exponentially in $r_{n}$. Moreover, inequality (\ref{bound on
excess risk}) together with the fact that $U_{n}\leq 1$ immediately implies
that 
\begin{equation}
E\left[ U_{n}\right] \leq \left( 1+j_{0}\right) e^{-\sigma r_{n}}+3\lambda s.
\label{bound on the mean excess risk}
\end{equation}%
By Condition \ref{lamda} and (\ref{estimated sparsity}), we can therefore
deduce that 
\begin{equation}
E\left[ U_{n}\right] =O\left( q\sqrt{n^{-1}\ln (p\vee n)}\right) ,
\label{theorem1-expectation}
\end{equation}%
which converges to zero whenever 
\begin{equation}
q^{2}\ln (p\vee n)=o(n).  \label{rate}
\end{equation}%
The rate condition (\ref{rate}) allows the case that 
\begin{equation}
\ln p=O(n^{\alpha })\text{ and }q=o(n^{1/2-\alpha /2})\text{ for }0<\alpha
<1.  \label{rate result}
\end{equation}%
In other words, the $\ell _{0}$-penalized ERM classification approach
is risk-consistent even when the dimension of the input feature space ($p$)
grows exponentially in sample size, provided that the number of truly
effective features ($q$) can only grow at a polynomial rate.

We shall provide some further remarks on the convergence rate result (\ref%
{theorem1-expectation}). Condition \ref{sparsity} implies that the space $%
\mathcal{B}_{q}$ contains the Bayes classifier $b^{\ast }$. Thus, if the
value of $q$ were known, one could performed classification via the $\ell
_{0}$-constrained ERM approach where the empirical risk $S_{n}(b)$ is
minimized with respect to $b\in \mathcal{B}_{q}$. The lower bound on the VC
dimension of the classifier space $\mathcal{B}_{q}$ grows at rate $O(q\ln
\left( p/q\right) )$ (\citet[][Lemma 1]{abramovich2017}). Hence, the rate $O(%
\sqrt{n^{-1}q\ln \left( p/q\right) })$ is the minimax optimal rate at which
the excess risk converges to zero under this constrained estimation approach
(\citet[][Theorem 14.5]{DGL:1996}). In view of this, suppose $p$ grows at a
polynomial or exponential rate in $n$. Then our rate result (\ref%
{theorem1-expectation}) is nearly oracle in the sense that, when $q$ grows
at rate $O(\ln n)$, the rate (\ref{theorem1-expectation}) remains close
within some $\ln n$ factor to the optimal rate attained under the case of
known $q$. Moreover, both rates coincide and reduce to $O(\sqrt{n^{-1}\ln p}%
) $ when the value of $q$ does not increase with the sample size.

\section{Computational Algorithms\label{Sec:MIO}}

While the ERM approach to binary classification is theoretically sound, its
implementation is computationally challenging and is known to be an \textit{NP%
} (Non-deterministic Polynomial time) hard problem (\citet{Johnson:78}). %
\citet{Florios:Skouras:08} developed a mixed integer optimization (MIO)
based computational method and provided numerical evidence demonstrating
effectiveness of the MIO approach to solving the ERM type optimization
problems. \citet{kitagawa2018} and \citet{mbakop2018} adopted the MIO
solution approach to solving the optimal treatment assignment problem which
is closely related to the ERM based classification problem. The MIO approach
is also useful for solving problems of variable selection through $\ell _{0}$%
-norm constraints. See \citet{bertsimas2016} and \citet{chen2018} who
proposed MIO based computational algorithms to solving the $\ell _{0}$-constrained regression and classification problems respectively.

Motivated by these previous works, we now present an MIO based computational
method for solving the $\ell _{0}$-penalized ERM problem. Given that $%
Y_{i}\in \{0,1\}$, solving the problem (\ref{penalized ERM}) amounts to
solving%
\begin{equation}
\min\nolimits_{\theta \in \Theta }\frac{1}{n}\sum\nolimits_{i=1}^{n}\left[
Y_{i}-\left( 2Y_{i}-1\right) 1\{X_{1i}+\widetilde{X}_{i}^{\prime }\theta
\geq 0\}\right] +\lambda \left\Vert \theta \right\Vert _{0}.
\label{ERM with penalty}
\end{equation}%
We assume that the parameter space $\Theta $ takes the form 
\begin{equation*}
\Theta =\prod\nolimits_{j=1}^{p}\left[ \underline{c}_{j},\overline{c}_{j}%
\right] ,
\end{equation*}%
where $\underline{c}_{j}$ and $\overline{c}_{j}$ are lower and upper
parameter bounds such that $-\infty <\underline{c}_{j}\leq \theta _{j}\leq 
\overline{c}_{j}<\infty $ for $j\in \{1,...,p\}$.

Our implementation builds on the method of mixed integer optimization.
Specifically, we note that the minimization problem (\ref{ERM with penalty})
can be equivalently reformulated as the following mixed integer linear
programming problem:%
\begin{align}
& \min_{\theta \in \mathbf{\Theta },d_{1},...,d_{n},e_{1},...,e_{p}}\frac{1}{%
n}\sum\nolimits_{i=1}^{n}\left[ Y_{i}-\left( 2Y_{i}-1\right) d_{i}\right]
+\lambda \sum\nolimits_{j=1}^{p}e_{j}  \label{MIO} \\
& \text{subject to}  \notag \\
& \left( d_{i}-1\right) M_{i}\leq X_{1i}+\widetilde{X}_{i}^{\prime }\theta
<d_{i}(M_{i}+\delta ),\text{ }i\in \{1,...n\},  \label{constraint on di} \\
& e_{j}\underline{\theta }_{j}\leq \theta _{j}\leq e_{j}\overline{\theta }%
_{j},\text{ }j\in \{1,...,p\},  \label{selection constraint} \\
& d_{i}\in \{0,1\},\text{ }i\in \{1,...n\},  \label{indicator di} \\
& e_{j}\in \{0,1\},\text{ }j\in \{1,...,p\},  \label{selection indicator ej}
\end{align}%
where $\delta $ is a given small and positive real scalar (e.g. $\delta
=10^{-6}$ as in our numerical study), and 
\begin{equation}
M_{i}\equiv \max_{\theta \in \Theta }\left\vert X_{1i}+\widetilde{X}%
_{i}^{\prime }\theta \right\vert \text{ for }i\in \{1,...,n\}.  \label{Mi}
\end{equation}

We now explain the equivalence between (\ref{ERM with penalty}) and (\ref%
{MIO}). Given $\theta $, the inequality constraints (\ref{constraint on di})
and the dichotomization constraints (\ref{indicator di}) enforce that $%
d_{i}=1\{X_{1i}+\widetilde{X}_{i}^{\prime }\theta \geq 0\}$ for $i\in
\{1,...n\}$. Moreover, the on-off constraints (\ref{selection constraint})
and (\ref{selection indicator ej}) ensure that, whenever $e_{j}=0$, the
value $\theta _{j}$ must also be zero so that the $j$th component of the
feature vector $\widetilde{X}$ is excluded in the resulting $\ell _{0}$-penalized ERM
classifier. The sum $\sum\nolimits_{j=1}^{p}e_{j}$ thus captures the number
of non-zero components of the vector $\theta $. As a result, both
minimization problems (\ref{ERM with penalty}) and (\ref{MIO}) are
equivalent. This equivalence enables us to employ modern MIO solvers to
solve for $\ell _{0}$-penalized ERM classifiers. For implementation, note that the
values $\left( M_{i}\right) _{i=1}^{n}$ in the inequality constraints (\ref%
{constraint on di}) can be computed by formulating the maximization problem
in (\ref{Mi}) as linear programming problems, which can be efficiently
solved by modern numerical solvers. Hence these values can be easily
computed and stored as inputs to the MILP problem (\ref{MIO}).

\section{Simulation Study\label{Sec:Simulation}}

In this section, we conduct simulation experiments to study the performance
of our approach. We consider a
simulation setup similar to that of \citet[][Section 5]{chen2018} and use
the following data generating design. Let $V=(V_{1},...,V_{p})$ be a
multivariate normal random vector with mean zero and covariance matrix $%
\Sigma $ with its element $\Sigma _{i,j}=\left( 0.25\right) ^{\left\vert
i-j\right\vert }$. The binary outcome is generated according to the
following specification:%
\begin{equation*}
Y=1\{X_{1}+\widetilde{X}^{\prime }\theta ^{\ast }\geq \sigma (X)\xi \},
\end{equation*}%
where $\theta ^{\ast }$ denotes the true data generating parameter value, $%
X=(X_{1},\widetilde{X})$ is a $(p+1)$ dimensional feature vector with $%
X_{1}=V_{1}$ and $\widetilde{X}=(1,V_{2},...,V_{p})$, and $\xi $ is a random
variate that is independent of $V$ and follows the standard logistic
distribution. The constant term in $\widetilde{X}$ is included to capture
the regression intercept. We set $\theta _{1}^{\ast }=0$ and $\theta
_{j}^{\ast }=0$ for $j\in \{3,...,p\}.$ The coefficient $\theta _{2}^{\ast }$
is chosen to be non-zero such that, among all the feature variables in $%
\widetilde{X}$, only the variable $\widetilde{X}_{2}=V_{2}$ is relevant in
the data generating processes (DGP). We consider the following two
specifications for $\theta _{2}^{\ast }$ and $\sigma (X)$: 
\begin{eqnarray*}
&&\text{DGP(i) : }\theta _{2}^{\ast }=-0.55\text{ and }\sigma (X)=0.2. \\
&&\text{DGP(ii) : }\theta _{2}^{\ast }=-1.85\text{ and }\sigma (X)=0.2\left(
1+2\left( V_{1}+V_{2}\right) ^{2}+\left( V_{1}+V_{2}\right) ^{4}\right) .
\end{eqnarray*}%
We used $100$ simulation repetitions in each Monte Carlo experiment. For
each simulation repetition, we generated a training sample of $n=100$
observations for estimating the coefficients $\theta $ and a validation
sample of $5000$ observations for evaluating the out-of-sample
classification performance. We considered simulation configurations with $%
p\in \{10,200\}$ to assess the classifier's performance in both the low and
high dimensional binary classification problems.

We specified the parameter space $\Theta $ to be $[-10,10]^{p}$ for the MIO
computation of the $\ell _{0}$-penalized ERM classifiers. Throughout this paper, we
used the MATLAB implementation of the Gurobi Optimizer to solve the MIO
problems (\ref{MIO}). Moreover, all numerical computations were done on a
desktop PC (Windows 7) equipped with 32 GB RAM and a CPU processor (Intel
i7-5930K) of 3.5 GHz. To reduce computation cost of solving the $\ell _{0}$%
-penalized ERM problems, we set the MIO solver time limit to be one
hour beyond which we forced the solver to stop early and used the best
discovered feasible solution to construct the resulting $\ell _{0}$-penalized ERM
classifier. For implementation, it remains to specify an exact form of the
penalty parameter $\lambda $ in (\ref{penalized ERM}). We set 
\begin{equation}
\lambda =v\left[ \ln \ln \left( p\vee n\right) \right] \sqrt{\ln \left(
p\vee n\right) /n},  \label{tuning parameter specification}
\end{equation}%
where $v$ is a tuning constant which remains to be calibrated. The form (\ref%
{tuning parameter specification}) implies that the value $c$ in Condition %
\ref{lamda} is taken to be $v\ln \ln (p\vee n)$, which will satisfy
inequality (\ref{condition on c}) and hence validate the probability bound (%
\ref{bound on the l0 norm of theta_hat}) when $p\vee n$ is sufficiently
large. Moreover, by the risk upper bound (\ref{bound on the mean excess risk}%
), the convergence rate result (\ref{theorem1-expectation}) continues to
hold up to a factor of $\ln \ln \left( p\vee n\right) $.

For practical applications, we recommend calibrating the tuning scalar $v$
via the method of cross validation. Yet, for this simulation study, we used
a simple heuristic rule and set%
\begin{eqnarray}
v &\equiv &h\left( 1-h\right) ,  \label{v} \\
h &\equiv &\min_{t\in \lbrack -10,10]}\frac{1}{n}\sum\nolimits_{i=1}^{n}1%
\left\{ Y_{i}\neq 1\left\{ X_{1i}+t\geq 0\right\} \right\} .  \label{h}
\end{eqnarray}%
The value $h$ in (\ref{h}) can also be computed via the MIO approach
by simply removing from the MIO problem (\ref{MIO}) the constraints (\ref%
{selection constraint}) and (\ref{selection indicator ej}) as well as the
binary controls $\left( e_{1},..,e_{p}\right) $ and the penalty part in the
objective function. 
This computation is much faster as it is concerned with
one-dimensional optimization. The rationale behind the choice (\ref{v}) is
as follows. Note that (\ref{h}) corresponds to an ERM classification using
classifier (\ref{predictor}) where $\widetilde{X}$ only consists of the
intercept term, and (\ref{v}) corresponds to an estimate of the variance of
the misclassification loss under such a simple classification rule.
Intuitively speaking, the value $v$ captures the variability of the
empirical risk under a parsimonious feature space specification. From the
bias and variance tradeoff perspective, when this variability is small, we
may as well increase the classifier flexibility by attaching a small penalty
in the penalized ERM procedure so as to induce a richer set of selected
features for classification.

Let logit\_lasso denote the $\ell _{1}$-penalized logistic regression
approach \citep[see
e.g.][]{friedman2010}. The logit\_lasso estimation approach is a
computationally attractive approach that can be used to estimate high
dimensional binary response models. We compared in simulations the
performance of our method to that of the
logit\_lasso approach. We used the MATLAB implementation of the well known 
\textbf{glmnet} computational algorithms \citep{Qian2013} for solving the
logit\_lasso estimation problems. 
We did not penalize the coefficient of the feature variable $X_{1}$ so that,
as in the simulations of the $\ell _{0}$-penalized ERM approach, this
variable would always be included in the resulting classifier constructed
under the logit\_lasso approach. We calibrated the lasso penalty parameter
value over a sequence of 100 values via the 10-fold cross validation
procedure. We used the default setup of \textbf{glmnet} for constructing
this tuning sequence among which we reported results based on the following
two choices, $\left\{ \lambda _{opt}^{lasso},\lambda _{1se}^{lasso}\right\} $%
, of the penalty parameter value. The value $\lambda _{opt}^{lasso}$ refers
to the lasso penalty parameter value that minimized the cross validated
misclassification risk, whereas $\lambda _{1se}^{lasso}$ denotes the largest
penalty parameter value whose corresponding cross validated
misclassification risk still falls within the one standard error of the
cross validated misclassification risk evaluated at $\lambda _{opt}^{lasso}$%
. The choice $\lambda _{1se}^{lasso}$ induces a more parsimonious estimating
model and is known as the "one-standard-error" rule, which is also commonly
employed in the statistical learning literature \citep{hastie2009}.

We considered the following performance measures. Let $\widehat{\theta }$
denote the estimated coefficients under a given classification approach.
For the logit\_lasso approach, we derived $\widehat{\theta }$ by dividing
the lasso-penalized logistic regression coefficients of the variables $%
\widetilde{X}$ by the magnitude of that of the variable $X_{1}$.
We can
easily deduce that $b_{\theta ^{\ast }}(X)=1\{X_{1}+\widetilde{X}^{\prime
}\theta ^{\ast }\geq 0\}$ is the Bayes classifier in this simulation design.
To assess the classification performance, we report the relative risk, which
is the ratio of the misclassification risk evaluated at the classifier $b_{%
\widehat{\theta }}$ over that evaluated at the Bayes classifier. In each
simulation repetition, we approximated the out-of-sample misclassification
risk using the generated validation sample. Let $in\_RR$ and $out\_RR$
respectively denote the average of in-sample and that of out-of-sample
relative risks over all the simulation repetitions.

We also examine the feature selection performance of the classification
method. We say that a feature variable $\widetilde{X}_{j}$ is effectively
selected if and only if the magnitude of $\widehat{\theta }_{j}$ is larger
than a small tolerance level (e.g., $10^{-6}$ as used in our numerical
study) which is distinct from zero in numerical computation. Let $Corr\_sel$
be the proportion of the variable $\widetilde{X}_{2}$ being effectively
selected. Let $Orac\_sel$ be the proportion of obtaining an oracle feature
selection outcome where the variable $\widetilde{X}_{2}$ was the only one
that was effectively selected among all the variables in $\widetilde{X}$.
Let $Num\_irrel$ denote the average number of effectively selected features
whose true DGP coefficients are zero.

\subsection{Simulation Results}

We now present in Tables \ref{tab1} and \ref{tab2} the simulation results
under the setups of DGP(i) and DGP(ii) respectively. From these two tables,
we find that, regarding the in-sample classification performance, our method outperformed the two logit\_lasso based
approaches across almost all the DGP configurations in the simulation. For
the out-of-sample classification performance, we see that the $\ell _{0}$%
-penalized ERM classifier dominated the logit\_lasso classifiers across all simulation
scenarios and this dominance was more evident in the high dimensional setup
with $p=200$.

\begin{table}[tbph]
\caption{Comparison of classification methods under DGP(i)}
\label{tab1}
\begin{center}
\begin{tabular}{c||ccc||ccc}
\hline\hline
& \multicolumn{3}{|c||}{${\small p=10}$} & \multicolumn{3}{|c}{${\small p=200%
}$} \\ \hline
{\small method} & $\ell _{0}${\small -ERM} & \multicolumn{2}{c||}{\small %
logit\_lasso} & $\ell _{0}${\small -ERM} & \multicolumn{2}{c}{\small %
logit\_lasso} \\ 
&  & ${\small \lambda }_{opt}^{lasso}$ & ${\small \lambda }_{1se}^{lasso}$ & 
& ${\small \lambda }_{opt}^{lasso}$ & ${\small \lambda }_{1se}^{lasso}$ \\ 
\hline
${\small Corr\_sel}$ & 0.98 & 1 & 0.99 & 0.94 & 0.99 & 0.87 \\ 
${\small Orac\_sel}$ & 0.95 & 0 & 0 & 0.83 & 0 & 0 \\ 
${\small Num\_irrel}$ & 0.03 & 3.05 & 1.45 & 0.15 & 7.26 & 2.62 \\ 
${\small in\_RR}$ & 0.828 & 0.870 & 1.134 & 0.778 & 0.843 & 1.237 \\ 
${\small out\_RR}$ & 1.094 & 1.168 & 1.304 & 1.139 & 1.313 & 1.471 \\ \hline
\end{tabular}%
\end{center}
\end{table}

\begin{table}[tbph]
\caption{Comparison of classification methods under DGP(ii)}
\label{tab2}
\begin{center}
\begin{tabular}{c||ccc||ccc}
\hline\hline
& \multicolumn{3}{|c||}{${\small p=10}$} & \multicolumn{3}{|c}{${\small p=200%
}$} \\ \hline
{\small method} & $\ell _{0}${\small -ERM} & \multicolumn{2}{c||}{\small %
logit\_lasso} & $\ell _{0}${\small -ERM} & \multicolumn{2}{c}{\small %
logit\_lasso} \\ 
&  & $\lambda _{opt}^{lasso}$ & ${\small \lambda }_{1se}^{lasso}$ &  & $%
\lambda _{opt}^{lasso}$ & ${\small \lambda }_{1se}^{lasso}$ \\ \hline
${\small Corr\_sel}$ & 0.91 & 1 & 0.91 & 0.83 & 0.95 & 0.84 \\ 
${\small Orac\_sel}$ & 0.91 & 0 & 0 & 0.82 & 0 & 0 \\ 
${\small Num\_irrel}$ & 0.01 & 2.99 & 1.48 & 0.05 & 9.86 & 2.72 \\ 
${\small in\_RR}$ & 0.893 & 0.969 & 1.095 & 0.884 & 0.804 & 1.069 \\ 
${\small out\_RR}$ & 1.071 & 1.160 & 1.248 & 1.103 & 1.271 & 1.289 \\ \hline
\end{tabular}%
\end{center}
\end{table}

Concerning the feature selection results, both Tables \ref{tab1} and \ref%
{tab2} indicate that all the three classification approaches had high $%
Corr\_sel$ rates and hence were effective for selecting the relevant
variable $\widetilde{X}_{2}$. However, the good performance in the $%
Corr\_sel $ criterion might just be a consequence of overfitting, which may
result in excessive selection of irrelevant variables and thus adversely
impact on the out-of-sample classification performance. From the results on
the $Num\_irrel $ performance measure, we note that the numbers of
irrelevant variables selected under the two logit\_lasso based approaches
remained quite large relatively to those under the $\ell _{0}$-penalized ERM approach even though all these approaches exhibited the effect
of shrinking the feature space dimension. In fact, we observe non-zero and
high values of $Orac\_Sel$ for the $\ell _{0}$-classifier across all the
simulation setups whereas the two logit\_lasso classifiers could not induce
any oracle variable selection outcome in the simulation. These feature
selection performance results help to explain that the risk performance
dominance of the $\ell _{0}$-penalized ERM approach could be observed
even in the DGP(i) simulations where the logistic regression model was
correctly specified.

\section{Conclusions\label{Sec:Conclusions}}

In this paper, we study the binary classification problem in a setting with
high dimensional vectors of features. We construct a binary classification
procedure by minimizing the empirical misclassification risk with a penalty
on the number of selected features. We establish a finite-sample probability
bound showing that this classification approach can yield a sparse solution
for feature selection with high probability. We also conduct non-asymptotic
analysis on the excess misclassification risk and establish its rate of
convergence. For implementation, we show that the penalized empirical risk
minimization problem can be solved via the method of mixed integer linear
programming.

There are a few topics one may consider as possible extensions. First, it might be fruitful to explore an $\ell_0$-penalized approach for regression  and other estimation problems.  Second,  our proposed method is suitable for training samples with small or moderate size. It would be a natural step to develop a divide-and-conquer algorithm for a large-scale problem \citep[see, e.g.,][]{Shi-et-al:2018}.
Third, our approach might be applicable for developing sparse policy learning rules \citep[see, e.g.,][]{athey2017efficient}.
These are topics for further research.

\appendix%

\section{Proofs of Theoretical Results\label{Sec:Appendix}}

We shall use the following lemma in the proofs of Theorems \ref{sparsity
bound} and \ref{convergence rate}.

\begin{lemma}
\label{tail bound}For all $\sigma >0$, there is a universal constant $%
M_{\sigma }$, which depends only on $\sigma $, such that%
\begin{equation}
P\left( \sup\limits_{b\in \mathcal{B}_{k}}\left\vert
S_{n}(b)-S(b)\right\vert >\sqrt{\frac{M_{\sigma }k\ln (p\vee n)}{n}}\right)
\leq e^{-\sigma k\ln (p\vee n)}  \label{bound on the tail probability}
\end{equation}%
for any integer $k\in \{1,...,p\}$ such that 
\begin{equation}
4\left( k+1\right) \ln \left( M_{\sigma }k\ln (p\vee n)\right) \leq k\ln
(p\vee n)+6\left( k+1\right) \ln 2.  \label{ineq1}
\end{equation}
\end{lemma}

The proof of Lemma \ref{tail bound} is a straightforward modification of
that of Theorem 1 of \citet{chen2018}. Hence we omit this proof here.

\subsection{Proof of Theorem \protect\ref{sparsity bound}}

\begin{proof}[Proof of Theorem \protect\ref{sparsity bound}]
We first prove the probability bound (\ref{bound on the l0 norm of theta_hat}%
). Let $\theta ^{\ast }\equiv \arg \inf_{\theta \in \Theta _{q}}S(b_{\theta
})$. Because $b^{\ast }\in \mathcal{B}_{q}$, it is straightforward to see
that%
\begin{eqnarray}
U_{n} &=&\left[ S(b_{\widehat{\theta }})-S_{n}(b_{\widehat{\theta }%
})-\lambda \left\Vert \widehat{\theta }\right\Vert _{0}\right] +\left[
S_{n}(b_{\widehat{\theta }})+\lambda \left\Vert \widehat{\theta }\right\Vert
_{0}-S(b^{\ast })\right]  \notag \\[1pt]
&\leq &\left[ S(b_{\widehat{\theta }})-S_{n}(b_{\widehat{\theta }})-\lambda
\left\Vert \widehat{\theta }\right\Vert _{0}\right] +\left[ S_{n}(b^{\ast
})+\lambda \left\Vert \theta ^{\ast }\right\Vert _{0}-S(b^{\ast })\right] 
\notag \\[1pt]
&\leq &\left\vert S_{n}(b_{\widehat{\theta }})-S(b_{\widehat{\theta }%
})\right\vert +\sup\nolimits_{b\in \mathcal{B}_{q}}\left\vert
S_{n}(b)-S(b)\right\vert +\lambda q-\lambda \left\Vert \widehat{\theta }%
\right\Vert _{0}.  \label{a}
\end{eqnarray}%
Since $U_{n}\geq 0$, it follows from (\ref{a}) that%
\begin{equation}
\left\Vert \widehat{\theta }\right\Vert _{0}\leq q+2\lambda ^{-1}\Delta
_{n}(\left\Vert \widehat{\theta }\right\Vert _{0}\vee q)
\label{theta_hat L0 norm upper bound}
\end{equation}%
where, for any $k\geq 0$, 
\begin{equation}
\Delta _{n}(k)\equiv \sup\nolimits_{b\in \mathcal{B}_{k}}\left\vert
S_{n}(b)-S(b)\right\vert .  \label{sup difference}
\end{equation}%
By construction, $0\leq S_{n}(b)\leq 1$ for any indicator function $b$. We
thus have that%
\begin{equation}
\lambda \left\Vert \widehat{\theta }\right\Vert _{0}\leq S_{n}(b_{\widehat{%
\theta }})+\lambda \left\Vert \widehat{\theta }\right\Vert _{0}\leq 1,
\label{estimated sparsity bound}
\end{equation}%
where the second inequality above follows by evaluating the objective
function in (\ref{penalized ERM}) at the $\theta $ vector whose components
are all zero. By (\ref{estimated sparsity}), we can deduce that 
\begin{equation}
\left\Vert \widehat{\theta }\right\Vert _{0}\leq p\wedge \left\lfloor
\lambda ^{-1}\right\rfloor .  \label{initial upper bound}
\end{equation}%
Note that, by (\ref{theta_hat L0 norm upper bound}), (\ref{initial upper
bound}) and (\ref{m0}), 
\begin{equation}
\left\Vert \widehat{\theta }\right\Vert _{0}\leq q+2\lambda ^{-1}\Delta
_{n}(m_{0}).  \label{theta_hat initial upper bound}
\end{equation}

Given $\sigma >0$, let $M_{\sigma }$ be the universal constant stated in
Lemma \ref{tail bound}. Let $\delta \equiv 2c^{-1}\sqrt{M_{\sigma }}$ and $c$
is the constant specified in Condition \ref{lamda}. Suppose $c$ is
sufficiently large such that inequality (\ref{condition on c}) holds. Since $%
\epsilon \in (0,1)$, we have that 
\begin{equation}
\delta \leq \epsilon \left( 1+\epsilon \right) ^{-1}<1.
\label{inequality for delta}
\end{equation}

For each positive integer $j$, let 
\begin{equation}
m_{j}\equiv q+\delta \sqrt{m_{j-1}}.  \label{m_j}
\end{equation}%
Given that $q\geq 1$, by (\ref{m_j}), we have that%
\begin{eqnarray*}
m_{j} &\leq &q+\delta m_{j-1} \\
&\leq &q+\delta q+\delta ^{2}m_{j-2} \\
&&.... \\
&\leq &q\frac{1-\delta ^{j}}{1-\delta }+\delta ^{j}m_{0}.
\end{eqnarray*}%
Therefore, by (\ref{j0}), we have that, for all $j\geq j_{0}$,%
\begin{equation}
m_{j}\leq q\frac{1-\delta ^{j}}{1-\delta }+\epsilon \leq s,
\label{bound on mj}
\end{equation}%
where the second inequality follows from (\ref{inequality for delta}).

By (\ref{theta_hat L0 norm upper bound}), we have that%
\begin{eqnarray*}
P\left( \left\Vert \widehat{\theta }\right\Vert _{0}\leq m_{j}\right) &\leq
&P\left( \left\Vert \widehat{\theta }\right\Vert _{0}\leq q+2\lambda
^{-1}\Delta _{n}(m_{j})\right) \\
&\leq &P\left( \left\Vert \widehat{\theta }\right\Vert _{0}\leq q+2\lambda
^{-1}\Delta _{n}(m_{j}),\Delta _{n}(m_{j})\leq \lambda c^{-1}\sqrt{M_{\sigma
}m_{j}}\right) \\
&&+P\left( \Delta _{n}(m_{j})>\lambda c^{-1}\sqrt{M_{\sigma }m_{j}}\right) \\
&\leq &P\left( \left\Vert \widehat{\theta }\right\Vert _{0}\leq
m_{j+1}\right) +P\left( \Delta _{n}(m_{j})>\lambda c^{-1}\sqrt{M_{\sigma
}m_{j}}\right) .
\end{eqnarray*}%
Hence 
\begin{equation}
P\left( \left\Vert \widehat{\theta }\right\Vert _{0}>m_{j+1}\right) \leq
P\left( \left\Vert \widehat{\theta }\right\Vert _{0}>m_{j}\right) +P\left(
\Delta _{n}(m_{j})>\lambda c^{-1}\sqrt{M_{\sigma }m_{j}}\right) .  \label{c}
\end{equation}%
By (\ref{initial upper bound}), $P\left( \left\Vert \widehat{\theta }%
\right\Vert _{0}>m_{0}\right) =0$. Using this fact and applying (\ref{c})
recursively, we have that%
\begin{equation}
P\left( \left\Vert \widehat{\theta }\right\Vert _{0}>m_{k}\right) \leq
\sum\nolimits_{j=0}^{k-1}P\left( \Delta _{n}(m_{j})>\lambda c^{-1}\sqrt{%
M_{\sigma }m_{j}}\right) .  \label{probability bound}
\end{equation}%
Therefore, result (\ref{bound on the l0 norm of theta_hat}) of Theorem \ref%
{sparsity bound} follows by noting that%
\begin{eqnarray}
&&P\left( \left\Vert \widehat{\theta }\right\Vert _{0}>s\right)  \notag \\
&\leq &\sum\nolimits_{i=0}^{j_{0}-1}P\left( \Delta _{n}(m_{i})>\lambda c^{-1}%
\sqrt{M_{\sigma }m_{i}}\right)  \label{a1} \\
&\leq &\sum\nolimits_{i=0}^{j_{0}-1}P\left( \Delta _{n}(\left\lfloor
m_{i}\right\rfloor \wedge p)>\lambda c^{-1}\sqrt{M_{\sigma }\left(
\left\lfloor m_{i}\right\rfloor \wedge p\right) }\right)  \label{a2} \\
&\leq &\sum\nolimits_{i=0}^{j_{0}-1}e^{-\sigma \left( \left\lfloor
m_{i}\right\rfloor \wedge p\right) \ln (p\vee n)}  \label{a3} \\
&\leq &j_{0}e^{-\sigma r_{n}},  \label{a4}
\end{eqnarray}%
where (\ref{a1}) follows from (\ref{bound on mj}) and (\ref{probability
bound}), (\ref{a2}) follows from the fact that $r\geq \left\lfloor
r\right\rfloor \wedge p$ and $\Delta _{n}(r)=\Delta _{n}(\left\lfloor
r\right\rfloor \wedge p)$ for all $r\geq 0$, and, because $q\leq m_{i}\leq
m_{0}\vee \left( \left( j_{0}-1\right) q+\sqrt{m_{0}}\right) $ for $i\in
\{0,1,2,...,j_{0}-1\}$, (\ref{a3}) follows from an application of Lemma \ref%
{tail bound}, where the value of $k$ in this lemma is taken over the range $%
\{q,q+1,...,$ $\left[ m_{0}\vee \left( \left( j_{0}-1\right) q+\left\lfloor 
\sqrt{m_{0}}\right\rfloor \right) \right] \wedge p\}$.
\end{proof}

\subsection{Proof of Theorem \protect\ref{convergence rate}}

\begin{proof}
We exploit the proof of Theorem \ref{sparsity bound} to show the probability
bound (\ref{bound on excess risk}). Specifically, using\ (\ref{a}) and (\ref%
{sup difference}) and noting that $s\geq q$, we have that%
\begin{equation*}
U_{n}\leq 2\Delta _{n}(\left\Vert \widehat{\theta }\right\Vert _{0}\vee
q)+\lambda s.
\end{equation*}%
Hence%
\begin{eqnarray}
P\left( U_{n}>3\lambda s\right) &\leq &P\left( \Delta _{n}(\left\Vert 
\widehat{\theta }\right\Vert _{0}\vee q)>\lambda s,\left\Vert \widehat{%
\theta }\right\Vert _{0}\leq s\right) +P\left( \left\Vert \widehat{\theta }%
\right\Vert _{0}>s\right)  \notag \\
&\leq &P\left( \Delta _{n}(\left\lfloor s\right\rfloor )>\lambda \sqrt{%
\left\lfloor s\right\rfloor }\right) +j_{0}e^{-\sigma r_{n}}  \label{b1} \\
&\leq &P\left( \Delta _{n}(\left\lfloor s\right\rfloor )>c^{-1}\sqrt{%
M_{\sigma }}\lambda \sqrt{\left\lfloor s\right\rfloor }\right)
+j_{0}e^{-\sigma r_{n}}  \label{b2} \\
&\leq &e^{-\sigma \left\lfloor s\right\rfloor \left( p\vee n\right)
}+j_{0}e^{-\sigma r_{n}}  \label{b3} \\
&\leq &\left( 1+j_{0}\right) e^{-\sigma r_{n}},  \label{b4}
\end{eqnarray}%
where (\ref{b1}) follows from the probability bound (\ref{bound on the l0
norm of theta_hat}) of Theorem \ref{sparsity bound}, (\ref{b2}) follows from
(\ref{condition on c}), which implies $c>\sqrt{M_{\sigma }}$, (\ref{b3})
follows from Lemma \ref{tail bound}, and (\ref{b4}) follows from the fact
that $q\leq s$.
\end{proof}

{\singlespacing
\bibliographystyle{econometrica}
\bibliography{BSBP}
}

\end{document}